\documentclass[conference]{IEEEtran}
\IEEEoverridecommandlockouts

\usepackage[ruled]{algorithm2e}
\usepackage[noend]{algpseudocode}

\usepackage{amsmath}
\usepackage{amssymb}
\usepackage{amsthm}
\usepackage{amsfonts}
\usepackage{nccmath}
\newtheorem{theorem}{Theorem}
\newtheorem{lemma}{Lemma}
\newtheorem{proposition}{Proposition}

\theoremstyle{definition}

\newtheorem{assumption}{Assumption}

\theoremstyle{remark}
\newtheorem{remark}{Remark}

\usepackage{mathtools}
\usepackage[margin=2cm]{geometry}
\usepackage{graphicx}
\graphicspath{ {./images/} }
\usepackage{caption}
\usepackage{subcaption}
\usepackage{color}
\usepackage{xcolor}
\usepackage{float}
\usepackage[noadjust]{cite}
\usepackage{balance}
\usepackage[noabbrev]{cleveref}
\crefformat{equation}{(#1)}
\usepackage{bbm}
\usepackage{bm}
\usepackage{comment}
\usepackage{arcs}

\usepackage{resizegather}

\usepackage{acronym}
\newacro{lqr}[LQR]{Linear Quadratic Regulator}
\newacro{slqr}[SLQR]{Structured Linear Quadratic Regulators}
\newacro{olqr}[OLQR]{Output-feedback Linear Quadratic Regulators}
\newacro{lqg}[LQG]{Linear Quadratic Gaussian}
\newacro{dare}[DARE]{Discrete-time Algebraic Riccati Equation}
\newacro{ouralgo}[RNPO]{Riemannian Newton-type Policy Optimization}
\newacro{po}[PO]{Policy Optimization}
\newacro{pg}[PG]{Projected Gradient}
\newacro{gf}[GF]{Gradient Flow}
\newacro{gd}[GD]{Gradient Descent}
\newacro{sgd}[SGD]{Stochastic Gradient Descent}
\newacro{sde}[SDE]{Stochastic Differential Equation}
\newacro{PBH}[PBH]{Popov-Belevitch-Hautus}

\usepackage{cite}
\usepackage{textcomp}
\def\BibTeX{{\rm B\kern-.05em{\sc i\kern-.025em b}\kern-.08em
		T\kern-.1667em\lower.7ex\hbox{E}\kern-.125emX}}

\usepackage[inline]{enumitem}

\def\bR{{\mathbb{R}}}

\newcommand*{\pdfset}[2]{%
  \leavevmode
  \pdfsave
  \pdfliteral{#1}%
  \rlap{#2}%
  \pdfrestore
  \phantom{#2}%
}

\newcommand{\tr}[1]{\ensuremath{\mathrm{tr}\left[ #1 \right]}}

\newcommand{\E}[2]{\ensuremath{\mathbb{E}_{#1}\left[ #2 \right]}}

\DeclareMathOperator{\lambdamax}{\text{$\overline{\lambda}$}}
\DeclareMathOperator{\lambdamin}{\text{$\underline{\lambda}$}}

\makeatletter
\newcommand{\algorithmfootnote}[2][\footnotesize]{%
  \let\old@algocf@finish\@algocf@finish
  \def\@algocf@finish{\old@algocf@finish
    \leavevmode\rlap{\begin{minipage}{\linewidth}
    #1#2
    \end{minipage}}%
  }%
}
\makeatother

\title{Duality-Based Stochastic Policy Optimization for Estimation with Unknown Noise Covariances}
\author{Shahriar Talebi, Amirhossein Taghvaei, Mehran Mesbahi 
\thanks{The authors are with the William E. Boeing Department of Aeronautics and Astronautics, University of Washington, Seattle, WA, USA. S. Talebi is also with the Department of Mathematics at the University of Washington. The research of the first and last authors has been supported by AFOSR grant FA9550-20-1-0053 and NSF grant ECCS-2149470. Emails: {\em\small shahriar@uw.edu}, {\em\small amirtag@uw.edu}, and {\em\small mesbahi@uw.edu}.}
}

\begin{document}

\maketitle
\begin{abstract}
Duality of control and estimation allows mapping recent advances in data-guided control to the estimation setup. This paper formalizes and utilizes such a mapping to consider learning the optimal (steady-state) Kalman gain when process and measurement noise statistics are unknown. Specifically, building on the duality between synthesizing optimal control and estimation gains, the filter design problem is formalized as direct policy learning. In this direction, the duality is used to extend existing theoretical guarantees of direct policy updates for \ac{lqr} to establish global convergence of the \ac{gd} algorithm for the estimation problem--while addressing subtle differences between the two synthesis problems. Subsequently, a \ac{sgd} approach is adopted to learn the optimal Kalman gain without the knowledge of noise covariances. The results are illustrated via several numerical examples.
\end{abstract}

\section{Introduction}
Duality of control and estimation provides an important relationship between two distinct synthesis problems in system theory~\cite{kalman1960general,kalman1960new,pearson1966duality}. In fact, duality has served as an effective bridge for developing theoretical and computational techniques in  one domain and then ``dualized" for use in the other. For instance, the stability proof of  the Kalman filter relies on the stabilizing feature of the optimal feedback gain for the dual \ac{lqr} optimal control problem~\cite[Ch. 9]{xiong2008introduction}.  The aim of this paper is to build on this dualization for the purpose of learning the optimal estimation policy via recent advances in data-driven algorithms for optimal control.

The setup that we consider is the estimation problem for a system with known linear dynamics and observation model, but unknown process and measurement noise covariances. The problem is to learn the optimal steady-state Kalman gain using a training data that consists of independent realizations of the observation signal. 
This problem has a long history in system theory, often examined in the context of adaptive Kalman filtering~\cite{mehra1970identification,mehra1972approaches,carew1973identification,belanger1974estimation,myers1976adaptive,tajima1978estimation}. 
The classical reference \cite{mehra1972approaches} includes a comprehensive summary of four solution approaches to this problem: Bayesian inference \cite{magill1965optimal,hilborn1969optimal,matisko2010noise}, Maximum likelihood~\cite{kashyap1970maximum,shumway1982approach}, covariance matching~\cite{myers1976adaptive}, and innovation correlation methods~\cite{mehra1970identification,carew1973identification}. 
The Bayesian and maximum likelihood setup are known to be computationally costly and covariance matching admits undesirable biases in practice. 
For these reasons, the innovation correlation based approaches are more popular and have been subject of more recent research~\cite{odelson2006new,aakesson2008generalized,dunik2009methods}.  
The  article~\cite{zhang2020identification} includes an excellent survey on this topic. 
Though relying strongly on the statistical assumptions on the model, these approaches do not provide {\em non-asymptotic} guarantees. 

On the optimal control side, there has been a number of recent advances in data-driven synthesis methods.
For example, first order methods have been adopted for state-feedback \ac{lqr} problems \cite{bu2019lqr, bu2020policy}. This direct policy optimization perspective has been
particularly effective as it has been shown that the \ac{lqr} cost is \emph{gradient dominant} \cite{fazel2018global}, allowing
the adoption and global convergence of first order methods for optimal feedback synthesis despite the non-convexity of the cost, when represented directly in terms of this policy.
Since then, \ac{po} using first order methods has been investigated for variants of \ac{lqr} problem, such as \ac{olqr} \cite{fatkhullin2020optimizing}, model-free setup \cite{mohammadi2021linear}, risk-constrained setup \cite{zhao2021global}, \ac{lqg} \cite{Tang2021analysis}, and recently, Riemannian constrained \ac{lqr} \cite{talebi2022policy}.

This paper aims to bring new insights to the classical estimation problem through the lens of control-estimation duality and utilizing recent advances in data-driven optimal control. In particular, we first argue that the optimal mean-squared error estimation problem is ``equivalent'' to an \ac{lqr} problem. This in turn, allows representing the problem of finding the optimal Kalman gain as that of optimal policy synthesis for the \ac{lqr} problem---under conditions distinct from what has been examined in the literature. In particular in this equivalent \ac{lqr} formulation, the cost parameters--relating to the noise covariances--are unknown and the covariance of initial state is not positive definite. 
By addressing these technical issues, we show how exploring this relationship leads to computational algorithms for learning optimal Kalman gain with non-asymptotic error guarantees.  

The rest of the paper is organized as follows. The estimation problem is formulated in \S\ref{sec:problem}, followed by the estimation-control duality relationship in \S\ref{sec:duality}. The theoretical analysis on policy optimization for the Kalman gain appears in \S\ref{sec:theory} while the proofs are deferred to \cite{talebi2022duality}. We propose an \ac{sgd} algorithm in \S\ref{sec:algorithms} with several numerical examples, followed by concluding remarks in \S\ref{sec:conclusion}.

 \section{Background and Problem Formulation}\label{sec:problem}
Consider the stochastic difference equation,
\begin{subequations} \label{eqn:sysdyn}
\begin{align}
    x(t+1) =& A x(t) + \xi(t),\\
    y(t) =& H x(t) + \omega(t),
\end{align}
\end{subequations}
where $x(t) \in \mathbb R^n$ is the state of the system, $y(t)\in \mathbb R^m$ is the observation, and $\{\xi(t)\}_{t\in \mathbb Z}$ and $\{\omega(t)\}_{t\in \mathbb Z}$ are the uncorrelated zero-mean process and measurement noise vectors, respectively, with the following covariances,
\[\E{}{\xi(t)\xi^\intercal(t)} = Q \in \bR^{n\times n}, \quad \E{}{\omega(t)\omega^\intercal(t)} = R \in \bR^{m\times m},\]
for some (possibly time-varying) positive (semi-)definite matrices $Q, R \succcurlyeq 0$.  Let $m_{0}$ and $P_0 \succcurlyeq 0$ denote the mean and covariance of the initial condition $x_0$. 

Now, let us fix a time horizon $T>0$ and define an estimation policy, denoted by $\mathcal P$, as a map that takes a history of the observation signal $\mathcal Y_T=\{y(0),y(1),\ldots,y(T-1)\}$ as an input and outputs an estimate of the state $x(T)$, denoted by $\hat{x}_{\mathcal P}(T)$. 
The filtering problem of interest is finding the estimation policy $\mathcal P$ that minimizes the mean-squared error, 
\begin{align}\label{eq:mse-x-P}
   \E{}{\|x(T)-\hat{x}_{\mathcal P}(T)\|^2}.
\end{align}
We make the following assumptions in our problem setup:
\begin{enumerate*}
    \item The matrices $A$ and $H$ are known, but the process and the measurement noise covariance matrices, $Q$ and $R$, are \emph{not} available.
    \item We have access to a training data-set that consists of independent realizations of the observation signal $\{y(t)\}_{t=0}^T$. However, ground-truth measurements of $x(T)$ is \emph{not} available.\footnote{
    This setting arises in various applications, such as aircraft wing dynamics, when approximate  or reduced-order models are employed, and the effect of unmodelled dynamics and disturbances are captured by the process noise.}    
\end{enumerate*}

It is not possible to directly minimize~\eqref{eq:mse-x-P} as the ground-truth measurement $x(T)$ is not available. Instead, we propose to minimize the mean-squared error in predicting the observation $y(T)$ as a  surrogate objective function. 
In particular, let us first define  $\hat{y}_{\mathcal P}(T)=H\hat{x}_{\mathcal P}(T)$ as the prediction for the observation $y(T)$. This is indeed a prediction since the estimate $\hat{x}_{\mathcal P}(T)$ depends only on the observations up to time $T-1$.  The optimization problem is now finding the estimation policy $\mathcal P$ that minimizes the mean-squared prediction error,
\begin{equation}\label{eq:min-pred-error-P}
    \mathcal J^{\text{est}}_T(\mathcal P) \coloneqq \E{}{\|y(T)-\hat{y}_{\mathcal P}(T) \|^2} .
\end{equation}
\subsubsection{Kalman filter}
Indeed, when $Q$ and $R$ are known, the solution is given by the celebrated Kalman filter algorithm~\cite{kalman1960new}. 
The algorithm involves an iterative procedure to update the estimate $\hat{x}(t)$ according to  
\begin{equation}\label{eq:KF-mean}
    \hat{x}(t+1) = A\hat{x}(t) + L(t)(y(t) - H \hat{x}(t)),~ \hat{x}(0) = m_{0},
\end{equation}
where $L(t):=AP(t)H^\intercal(HP(t)H^\intercal + R)^{-1}$ is the Kalman gain, and $P(t):=\mathbb E[(x(t) - \hat{x}(t))(x(t) - \hat{x}(t))^\intercal]$ is the error covariance matrix that satisfies the Ricatti equation,
\begin{align*}
    P(t+1) = (A-L(t)H)P(t)A^\intercal + Q,\quad P(t_0) = P_0. 
\end{align*}
Note that the update law presented here combines the information and dynamic update steps of the Kalman filter. 

It is known that $P(t)$ converges to an steady-state value $P_\infty$ when the pair $(A,H)$ is observable and the pair $(A,Q^{\frac{1}{2}})$ is controllable~\cite{kwakernaak1969linear,lewis1986optimal}. In such a case, the gain converges to $L_\infty :=AP_\infty H^\intercal(HP_\infty H^\intercal + R)^{-1}$, the so-called steady-state Kalman gain. It is a common practice to evaluate the steady-state Kalman gain $L_\infty$ offline and use it, instead of $L(t)$, to update the estimate in real-time.   
\subsubsection{Learning the optimal Kalman gain}
Inspired by the structure of the Kalman filter, we consider restriction of the estimation policies $\mathcal P$ to those realized with a constant gain. In particular, we define the estimate $\hat{x}_L(T)$ as one given by the Kalman filter at time $T$ realized by the constant gain $L$. Rolling out the update law~\eqref{eq:KF-mean} for $t=0$ to $t=T-1$, and replacing $L(t)$ with $L$, leads to the following expression for the estimate $\hat{x}_L(T)$ as a function of $L$, 
\begin{align}\label{eq:estimate-x-L}\textstyle
    \hat{x}_L(T) = A_L^Tm_0 + \sum_{t=0}^{T-1}A_L^{T-t-1} L y(t),
\end{align}
where $A_L\coloneqq A-LH$.
Note that this estimate does not require knowledge of the matrices $Q$ or $R$. By considering $\hat{y}_L(T):=H\hat{x}_L(T)$, the problem is now finding the optimal gain $L$ that minimizes the mean-squared prediction error
\begin{equation}\label{eq:mse-x-L}
  J^{\text{est}}_T(L):=\E{}{\|y(T)-\hat{y}_{L}(T)\|^2}.
\end{equation}

Numerically, this problem falls into the realm of stochastic optimization and can be solved by algorithms such as \acf{sgd}. Such an algorithm would require accessing independent realizations of the observation signal. An algorithm that utilizes such realizations is presented in \S\ref{sec:algorithms}.
Theoretically, however, it is not yet clear if this optimization problem is well-posed  and admits a unique minimizer. This is the subject of \S\ref{sec:theory}, where certain properties of the objective function, such as its gradient dominance and smoothness, are established. These theoretical results are then used to analyze first-order optimization algorithms and provide stability guarantees of the estimation policy iterates. The results are based on the duality relationship between estimation and control that is presented next.

\section{Estimation-Control Duality Relationship}\label{sec:duality}
We use the duality framework, as described in~\cite[Ch.7.5]{aastrom2012introduction}, to relate the problem of learning the optimal estimation policy to that of learning the optimal control policy for an \ac{lqr} problem. In order to do so, we introduce the adjoint  system:
\begin{equation}\label{eqn:adjdyn}
    z(t) = A^\intercal z(t+1) - H^\intercal u(t+1),
\end{equation}
where $z(t) \in \mathbb R^n$ is the adjoint state and $\mathcal U_T:=\{u(1),\ldots,u(T)\} \in \mathbb R^{mT}$ are the control variables (dual to the observation signal $\mathcal Y_T$). The adjoint state is initialized at $z(T)=a \in \mathbb R^n$ and simulated  \textit{backward in time} starting with $t= T-1$.
We now formalize a relationship between estimation policies for the  system~\eqref{eqn:sysdyn} and control policies for the adjoint system~\eqref{eqn:adjdyn}. 
Consider estimation policies  that are linear functions of the observation history $\mathcal Y_T\in \mathbb R^{mT}$ and the initial mean vector $m_0 \in \mathbb R^n$. We characterize such policies with a linear map $\mathcal L:\mathbb R^{mT+n} \to \mathbb R^n$ and let the estimate  $\hat{x}_{\mathcal L}(T) := \mathcal L(m_0,\mathcal Y_T)$. The adjoint of this linear map, denoted by $\mathcal L^\dagger: \mathbb R^n \to \mathbb R^{mT+n}$, is used to define a control policy for the adjoint system~\eqref{eqn:adjdyn}. In particular, the adjoint map takes $a\in \mathbb R^n$ as input and outputs  $\mathcal L^\dagger(a)=\{b,u(1),\ldots,u(T)\} \in \mathbb R^{mT+n}$. 
This relationship can be depicted as,
\begin{align*}
    \{m_0,y(0),\ldots,y(T-1)\} \quad &\overset{\mathcal L}{\longrightarrow} \quad 
    \hat{x}_{\mathcal L}(T)\\
     \{b,u(1),\ldots,u(T)\} \quad &\overset{\mathcal L^\dagger}  {\longleftarrow} \quad a
\end{align*}
Note that
\(\langle a,\mathcal L(m_0,\mathcal Y_T)\rangle_{\mathbb R^n} = \langle \mathcal L^\dagger(a),(m_0,\mathcal Y_T)\rangle_{\mathbb R^{mT+n}},\)
so
\begin{align}\label{eq:adjoint}\textstyle
         b^\intercal m_0 + \sum_{t=0}^{T-1}u(t+1)^\intercal y(t) &= a^\intercal \hat{x}_{\mathcal L}(T).
\end{align}

The following proposition relates the mean-squared error for a linear estimation policy, to the following \ac{lqr} cost:
\begin{multline}
     \mathcal J_T^{\text{LQR}}(a,\{b,\mathcal U_T\})  := [z^\intercal(0) m_0 - b^\intercal m_0]^2
    \\ \textstyle +z^\intercal(0) P_0 z(0) + \sum_{t=1}^{T} \left[ z^\intercal(t) Q z(t) + u^\intercal(t) R u(t) \right].\label{eq:lqr}
\end{multline}

\begin{proposition}\label{prop:duality}
Consider the estimation problem for the system~\eqref{eqn:sysdyn} and the \ac{lqr} problem~\eqref{eq:lqr} subject to the adjoint dynamics~\eqref{eqn:adjdyn}. For each estimation policy $\hat{x}_{\mathcal L}(T)= \mathcal L(m_0,\mathcal Y_T)$, with a linear map $\mathcal L$, and for any $a \in \mathbb R^n$ we have the identity 
\begin{equation*}
    \E{}{|a^\intercal x(T)-a^\intercal \hat x_{\mathcal L}(T)|^2} =  \mathcal J_T^{\text{LQR}}(a,\mathcal L^\dagger(a)).
\end{equation*}
Furthermore, the prediction error as in \cref{eq:mse-x-L} satisfies
\begin{equation*}
    J^\text{est}_T(L) = \sum_{i=1}^m \mathcal J_T^{\text{LQR}}(H_i,\mathcal L^\dagger(H_i)) + \tr{R},
\end{equation*}
where $\hat{y}_{\mathcal L}(T) := H \hat{x}_{\mathcal L}(T)$ and $H_i^\intercal \in \mathbb R^n$ is the $i$-th row of the $m\times n$ matrix $H$ for $i=1,\ldots,m$.
\end{proposition}

\begin{remark}
The duality is also true in the continuous-time setting where the estimation problem is related to a continuous-time \ac{lqr} problem.    Recent extensions to the nonlinear setting appears in~\cite{kim2019lagrangian} with a comprehensive study in~\cite{kim2022duality}.
This duality is different than the maximum likelihood approach which involves an optimal control problem over the original dynamics instead of the adjoint system.
\end{remark}

\subsubsection{Duality in the constant control gain regime}
In this section, we use the aforementioned duality relationship to show that the estimation policy with constant gain is dual to the  control policy with constant feedback gain. This result is then used to obtain an explicit formula for the objective function~\eqref{eq:mse-x-L}.

Consider the adjoint system~\eqref{eqn:adjdyn} with the linear feedback law $u(t)=L^\intercal z(t)$.
Then, 
\begin{equation}\label{eq:z-sol}
z(t)=(A_L^\intercal)^{T-t}a,\quad \text{for}\quad t=0,1,\ldots,T.
\end{equation}
Therefore, as a function of $a$, $u(t) = L^\intercal(A_L^\intercal)^{T-t}a$. Moreover, for this choice of control, the optimal $b=z(0)= (A_L^\intercal)^{T}a$. These relationships are used to identify the control policy 
\(
  \mathcal L^\dagger (a) 
  =((A_L^\intercal)^{T}a, L^\intercal(A_L^\intercal)^{T-1}a,\ldots,L^\intercal  a).
\)
This control policy corresponds to an estimation policy by the adjoint relationship~\eqref{eq:adjoint}:
\begin{equation*}\textstyle
    a^\intercal \hat{x}_{\mathcal L}(T) = a^\intercal A_L^Tm_0 + \sum_{t=0}^{T-1} a^\intercal A_L^{T-t-1} L y(t),\quad \forall a \in \mathbb R^n.
\end{equation*}
As this relationship holds for all $a\in \mathbb R^n$, we have,
\begin{align*}\textstyle
  \hat{x}_{\mathcal L}(T) = A_L^Tm_0 + \sum_{t=0}^{T-1} A_L^{T-t-1} L y(t),
\end{align*}
that coincides with the Kalman filter estimate with constant gain $L$ given by the formula~\eqref{eq:estimate-x-L}. Therefore, the adjoint relationship~\eqref{eq:adjoint} relates the control policy with constant gain $L^\intercal$ to the Kalman filter with the constant gain $L$. 

Next, we use this relationship to evaluate the mean-squared prediction error~\eqref{eq:mse-x-L}.  
Denote by $J_T^{\text{LQR}}(a,L^\intercal)$ as the \ac{lqr} cost~\eqref{eq:lqr} associated with the control policy with constant gain $L^\intercal$ and  $b=z(0)$. Then, from the explicit formula for $z(t)$ and $u(t)$ above, we have, 
\begin{align*}
    J_T^{\text{LQR}}(a,L^\intercal) = a^\intercal X_T(L)a,
\end{align*}
where 
\begin{gather*}
    X_T(L) \coloneqq A_L^T P_0(A_L^\intercal)^T   + \sum_{t=1}^T  A_L^{T-t} (Q+LRL^\intercal)(A_L^\intercal)^{T-t}.
\end{gather*}
 Therefore, by the second claim in \Cref{prop:duality}, the mean-squared prediction error \cref{eq:mse-x-L} becomes,
 \begin{multline*}
    J^\text{est}_T(L)- \tr{R} = \sum_{i=1}^m J_T^\text{LQR}(H_i,L^\intercal)
    = \tr{X_T(L)H^\intercal H},
 \end{multline*}
where we have used the cyclic permutation property of the trace and the identity $H^{\intercal} H = \sum_{i=1}^m  H_i H_i^\intercal$.

\subsubsection{Duality in steady-state regime}
Define the set of Schur stabilizing gains
\[\mathcal{S} \coloneqq \{L \in \bR^{n\times m}: \rho(A-LH) < 1\}.\]
For any $L \in \mathcal S$, in the steady-state limit as $T \to \infty$: 
\(X_T(L) \to X_{\infty}(L) \coloneqq \sum_{t=0}^{\infty} \left(A_L\right)^{t} \left(  Q + L R L^\intercal \right)\left(A_L^\intercal\right)^{t}.\)
The limit coincides with the unique solution $X$ of the discrete Lyapunov equation
\(X = A_L X A_L^\intercal + Q + L R L^\intercal,\)
which exists as $\rho(A_L) < 1$.
Therefore, the steady-state limit of the mean-squared prediction error assumes the form,
\[J(L) \coloneqq \lim_{T \to \infty} J_{T}^\text{est}(L) = \tr{X_{\infty}(L) H^\intercal H} + \tr{R}.\]
Given the steady-state limit, we formally analyze the following constrained optimization problem:
\begin{align}\label{eq:opt-time-indepen}
    \min_{L \in \mathcal{S}} \; &\leftarrow J(L) = \tr{X_{(L)} H^\intercal H} + \tr{R},\\
              &\text{s.t.} \quad X_{(L)} = A_L X_{(L)} A_L^\intercal + Q + L R L^\intercal .\nonumber
\end{align}
\begin{remark}\label{remark:duality}
Note that the latter problem is technically the dual of the optimal \ac{lqr} problem as formulated in \cite{bu2019lqr} by relating  $A \leftrightarrow A^\intercal$, $-H \leftrightarrow B^\intercal$, $L \leftrightarrow K^\intercal$, and $H^\intercal H \leftrightarrow \Sigma$. However, one main difference here is that the matrices $Q$ and $R$ are unknown, and the $H^\intercal H$ may \emph{not} be positive definite, for example, due to rank deficiency in $H$ specially whenever $m < n$. Thus, in general, the cost function $J(L)$ is not necessarily coercive in $L$, which can drastcially effect the optimization landscape. For the same reason, in contrast to the \ac{lqr} case \cite{fazel2018global,bu2019lqr},
the gradient dominant property of $J(L)$ is not clear in the filtering setup. In the next section, we show that such issues can be avoided as long as the pair $(A,H)$ is observable. 
\end{remark}

\section{Theoretical analysis} \label{sec:theory}
In this section, we provide theoretical analysis of the proposed optimization problem~\eqref{eq:opt-time-indepen}. 
The following  lemma is useful for
our subsequent analysis which is a direct consequence of duality described in \Cref{remark:duality}, Lemmas 3.5 and 3.6 in \cite{bu2019lqr}, and the fact that the spectrum of a matrix remains unchanged under the transpose operation.
\begin{lemma}\label{lem:topology}
The set of Schur stabilizing gains $\mathcal{S}$ is regular open, contractible, and unbounded when $m\geq 2$ and the boundary $\partial \mathcal{S}$ coincides with the set $\{L \in \bR^{n\times m}: \rho(A-LH) = 1\}$. Furthermore, $J(.)$ is real analytic on $\mathcal{S}$ whenever $Q$ and $R$ are time-independent.
\end{lemma}

\subsubsection{Coercive property}
Next, we provide sufficient conditions to recover the coercive property of $J(.)$ which resembles Lemma 3.7 in \cite{bu2019lqr}, but extended for the time-varying cost parameters $Q$ and $R$.

\begin{proposition}\label{prop:coercive}
Suppose the pair $(A,H)$ is observable, and $Q$ and $R$ are lower bounded uniformly in time with some positive definite matrices. Then, the function $J(.):\mathcal{S} \to \bR$ is coercive, i.e., for any sequence $\{L_k\} \in \mathcal{S}$,
\[\text{ if~~} L_k \to \partial\mathcal{S} \text{~~or~~} \|L_k\| \to \infty \text{~~then~~} J(L) \to \infty.\]
Furthermore, for any $\alpha >0$, the sublevel set $\mathcal{S}_\alpha \coloneqq \{L \in \bR^{n\times m}: J(L) \leq \alpha\}$ is compact and contained in $\mathcal{S}$ whenever $Q$ and $R$ are time-independent.
\end{proposition}

\begin{remark}
This approach recovers the claimed coercivity also in the control setting with weaker assumptions. In particular, using this result, one can replace the positive definite condition on the covariance of the initial condition in \cite{bu2019lqr}, i.e., $\Sigma \succ 0$, with just the controllability of $(A, \Sigma^{1/2})$. 
\end{remark}

\subsubsection{Gradient dominance property}
Next, we establish the gradient dominance property which resembles Lemma 3.12 in \cite{bu2019lqr}. While our approach utilizes a similar proof technique, this property is not trivial in this case as $H^\intercal H$ may not be positive definite. This, apparently minor issue, hinders establishing the gradient dominated property globally. However, we are able to recover this property on every sublevel sets of $J(L)$ which is sufficient for the subsequent convergence analysis. 

Before presenting the result, we compute the gradient of $J(L)$ to characterize its global minimizer and consider the following simplifying assumption for the rest of the analysis.
\begin{assumption}
Suppose $(A,H)$ is observable and the covariance matrices $Q\succ 0$ and $R \succ 0$ are time-independent.
\end{assumption}
The explicit gradient formula for $J$ takes the form,
\begin{gather*}
\begin{aligned}
    \nabla  J(L) 
    =& 2Y_{(L)} \left[-LR + A_L X_{(L)} H^{\intercal} \right],
\end{aligned}
\end{gather*}
where $Y_{(L)}$ is the unique solution of
\(Y = A_L^\intercal Y A_L + H^\intercal H.\)
While the derivation appears in \cite{talebi2022duality}, note that the expression for the gradient is consistent with Proposition 3.8 in \cite{bu2019lqr} after applying the duality relationship explained in \Cref{remark:duality}.

We also characterize the global minimizer $L^* = \arg\min_{L \in \mathcal{S}} J(L)$. The domain $\mathcal{S}$ is non-empty whenever $(A,H)$ is observable. 
Thus, by continuity of $L \to J(L)$, there exists some finite $\alpha > 0$ such that the sublevel set $\mathcal{S}_\alpha$ is non-empty and compact. Therefore, the minimizer  is an interior point and thus must satisfy the first-order optimality condition $\nabla J(L^*) = 0$. Moreover, by coercivity, the minimizer is stabilizing and unique satisfying,
\[L^* = A X^* H^\intercal \left(R + H X^* H^\intercal\right)^{-1},\]
with $X^*$ being the unique solution of 
\begin{equation}\label{eqn:Xstar}
    X^* = A_{L^*} X^* A_{L^*}^\intercal + Q + L^* R (L^*)^\intercal.
\end{equation}
As expected, the global minimizer $L^*$ is equal to the steady-state Kalman gain, but explicitly dependent on the noise covariances $Q$ and $R$.

\begin{proposition}\label{prop:graddom}
Let $L^*$ be the unique optimizer of $J(L)$ over $\mathcal{S}$ and consider any non-empty sublevel set $\mathcal{S}_\alpha$ for some $\alpha>0$. Then, the function $J(.):\mathcal{S}_\alpha \to \bR$ satisfies
\begin{equation*}
 c_1 [J(L) - J(L^*)] + c_2 \|L-L^*\|_F^2 \leq  \langle \nabla J(L), \nabla J(L) \rangle,
\end{equation*}
\[c_3 \|L - L^*\|_F^2 \leq J(L) - J(L^*),\]
for some positive constants $c_1 = c_1(\alpha)> 0$, $c_2 = c_2(\alpha) > 0$ and $c_3 = c_3(\alpha) > 0$ that are independent of $L$.
\end{proposition}
\begin{remark}
The proposition above implies that $J(.)$ is gradient dominated on $\mathcal{S}_\alpha$, i.e., for any $L \in \mathcal{S}_\alpha$ we have
\[\textstyle J(L) - J(L^*) \leq  \frac{1}{c_1(\alpha)} \langle \nabla J(L), \nabla J(L) \rangle.\]
Note that the first inequality characterizes the dominance gap in terms of the iterate error from the optimality. This is useful in obtaining the iterate convergence results in the next section where we analyze first-order methods in order to solve the minimization problem~\eqref{eq:opt-time-indepen}.
\end{remark}

\subsection{\acf{gd}}
Here, we consider the \ac{gd} policy update:
\begin{flalign*} 
\text{[\ac{gd}]} \qquad\qquad\quad L_{k+1} = L_k - \eta_k \nabla J(L_k), &&
\end{flalign*}
for $k \in \mathbb Z$ and a positive  stepsize $\eta_k$.
As a direct consequence of \Cref{prop:graddom}, we can guarantee convergence for the \acf{gf} algorithm (see \cite{talebi2022duality} for details). But then, establishing convergence for \ac{gd} relies on carefully choosing the stepsize $\eta_k$, and bounding the rate of change of $\nabla J(L)$---at least on each sublevel set. So, the following lemma provides a Lipschitz bound for $\nabla J(L)$ on every sublevel set. This results resembles its ``dual'' counterpart in \cite[Lemma 7.9]{bu2019lqr}, however, it is \emph{not} implied directly by the duality argument as $H^\intercal H$ may not be positive definite.

\begin{lemma}\label{lem:lipschitz}
Consider any (non-empty) sublevel set $\mathcal{S}_\alpha$ for some $\alpha>0$. Then,
\[\|\nabla J(L_1) - \nabla J(L_2)\|_F \leq \ell\; \|L_1 - L_2\|_F, \quad \forall L_1,L_2 \in \mathcal{S}_\alpha,\]
for some positive constant $\ell = \ell(\alpha)>0$ that is independent of both $L_1$ and $L_2$.
\end{lemma}

In what follows, we establish linear convergence of the \ac{gd} algorithm. Our convergence result only depends on the value of $\alpha$ for the initial sublevel set $\mathcal{S}_\alpha$ that contains $L_0$. Note that our proof technique is distinct from those in \cite{bu2019lqr} and \cite{mohammadi2021convergence}; nonetheless, it involves a similar argument using the gradient dominance property of $J$.
\begin{theorem}\label{thm:graddescent }
Consider any sublevel set $\mathcal{S}_\alpha$ for some $\alpha >0$. Then, for any initial policy $L_0 \in \mathcal{S}_\alpha$, the \ac{gd} updates with any fixed stepsize $\eta_k = \eta \in (0, 1/\ell(\alpha)]$ converges to optimality at a linear rate of $1- \eta c_1(\alpha)/2$ (in both the function value and the policy iterate). In particular, we have
\[J(L_k) - J(L^*) \leq [\alpha - J(L^*)] (1- \eta c_1(\alpha)/2)^k,\]
and 
\(\|L_k - L^*\|_F^2 \leq \left[\frac{\alpha - J(L^*)}{c_3(\alpha)}\right] (1- \eta c_1(\alpha)/2)^k,\)
with $c_1(\alpha)$ and $c_3(\alpha)$ as defined in \Cref{prop:graddom}.
\end{theorem}

\section{Algorithms and Numerical Simulations}
\label{sec:algorithms}
In this section, we discuss numerical algorithms in order to solve the minimization problem~\eqref{eq:opt-time-indepen}. Note that, it is not possible to implement the gradient-descent algorithm because evaluating the gradient involves the  noise covariance matrices $Q$ and $R$,  assumed to be unknown.  
Instead, here we explore alternative approaches to recover the gradient information from the data at hand. 

\subsubsection{\acf{sgd}}
Herein, we allow a variable initial time $t_0$ (instead of just $t_0=0$) for the system~\eqref{eqn:sysdyn} and use  $\mathcal Y_{\{t_0:T\}}:=\{y(t_0),y(t_0+1),\ldots, y(T-1)\}$ to denote the measurement time-span. Using this notation, the statistical steady-state can be equivalently considered as the limit $t_0 \to -\infty$  with fixed $T$. 

Recall that any choice of $L \in \mathcal{S}$ corresponds to a filtering strategy that outputs a prediction $\hat{y}_L(T)$, which with the variable initial time $t_0$, is given by  
\begin{equation*}\textstyle
  \hat{y}_L(T) = H A_L^{T-t_0} m_0 +\sum_{t=t_0}^{T-1} H  A_L^{T-t-1} L y(t).
\end{equation*}
Also, let $e_{\{t_0:T\}}(L) \coloneqq y(T) - \hat{y}_L(T)$ denote the incurred error corresponding to this filtering strategy and let   
\[\varepsilon(L,\mathcal Y_{\{t_0:T\}}) \coloneqq \|e_{\{t_0:T\}}(L)\|^2,\]
denote the squared-norm of the error, where the dependence on the measurement sequence $\mathcal Y_{\{t_0:T\}}$ is explicitly specified. 

The optimization objective function is then to minimize the expectation of the squared-norm of the error over all possible random measurement sequences:  
\begin{equation*}
    J_{\{t_0:T\}}(L):=\E{}{\varepsilon(L,\mathcal Y_{\{t_0:T\}})};
\end{equation*}
at the steady-state, we obtain
\(\lim_{t_0\to -\infty}J_{\{t_0:T\}}(L) = J(L).\)

The \ac{sgd} algorithm aims to solve this optimization problem by replacing the gradient, in the \ac{gd} update, with an unbiased estimate of the gradient in terms of samples from the measurement sequence. In particular, assuming access to an oracle that produces independent realization of the measurement sequence, say $M$ randomly selected measurements  $\{\mathcal Y_{[t_0,T]}^i\}_{i=1}^M$, the gradient can be approximated according to
\begin{equation*}\textstyle
    \nabla J_{\{t_0:T\}}(L) \approx \frac{1}{M}\sum_{i=1}^M \nabla_L \varepsilon(L,\mathcal Y^i_{\{t_0:T\}}).
\end{equation*}
This forms an unbiased estimate of the gradient, i.e.,
\begin{equation*}\textstyle
     \E{}{\frac{1}{M}\sum_{i=1}^M \nabla_L \varepsilon(L,\mathcal Y^i_{\{t_0:T\}})} =  \nabla J_{\{t_0:T\}}(L),
\end{equation*}
with variance that converges to zero with the rate $O(\frac{1}{M})$ as the number of samples increase. The number $M$ is referred to as the batch-size. 

Using the stochastic estimation of the gradient, the algorithm proceeds as follows: we let,
\begin{flalign*} \textstyle
\text{[\ac{sgd}]} \qquad L_{k+1} = L_k - \frac{\eta_k}{M}\sum_{i=1}^M  \nabla_L \varepsilon(L,\mathcal Y^i_{\{t_0:T\}}), &&
\end{flalign*}
for $k \in \mathbb Z$, where $\eta_k>0$ is the step-size and $\{\mathcal Y^i_{\{t_0:T\}}\}$ represent $M$ fresh realizations of the measurement sequence. 

Although the convergence of the \ac{sgd} algorithm is expected to follow similar to the \ac{gd} algorithm under the  gradient dominance condition and Lipschitz property, the analysis becomes complicated due to the possibility of the iterated gain $L_k$ leaving the sub-level sets. It is expected that a convergence guarantee would hold under high-probability due to concentration of the gradient estimate around the true gradient. Complete analysis in this direction will be presented in our subsequent work. 

Finally, for implementation purposes, we compute the  gradient estimate  explicitly in terms of the measurement sequence and the filtering policy $L$.
\begin{lemma}\label{lem:grad-approx}
Given $L \in \mathcal{S}$ and a sequence of measurements $\mathcal Y= \{y(t)\}_{-\infty}^T$, we have,
\begin{multline*}
    \nabla_L \varepsilon(L,\mathcal Y)= -2\sum_{t=0}^{\infty} (A_L^\intercal)^{t} H^\intercal e_{T}(L) y^\intercal(T-t-1)    \\
    +2\sum_{t=1}^{\infty}\sum_{k=1}^{t} (A_L^\intercal)^{t-k} H^\intercal e_{T}(L) y^\intercal(T-t-1) L^\intercal (A_L^\intercal)^{k-1} H^\intercal .
\end{multline*}
\end{lemma}

\begin{remark}
Computing the gradient above only requires the knowledge of the system parameters $A$ and $H$, and does \emph{not} require the noise covariance information $Q$ and $R$.
\end{remark}

\subsubsection{Numerical Simulations}
Herein, we showcase the application of the developed theory for improving the estimation policy for an LTI system. Specifically, we consider an undamped mass-spring system  with known parameters $(A,H)$ with $n=2$ and $m=1$. In the hindsight, we consider a variance of $0.1$ for each state dynamic noise, a state covariance of $0.05$ and a variance of $0.1$ for the observation noise. Assuming a trajectory of length $T$ at every iteration, the approximate gradient is obtained as in \Cref{lem:grad-approx}, only requiring an output data sequence collected from the system in \cref{eqn:sysdyn}. Then, the progress of policy updates using the \ac{sgd} algorithm for different values of trajectory length $T$ and batch size $M$ are depicted in \Cref{fig:sim} where each figure shows statistics over 20 rounds of simulation. The figure demonstrates a ``sublinear rate'' of convergence which is expected as every update only relies on an approximation of the gradient---in contrast to the linear convergence established for \ac{gd}. Finally, \Cref{fig:3} demonstrates also the convergence in the Kalman gain as predicted by the properties of $J$ studied in \S\ref{sec:theory} (see \Cref{prop:graddom}).

\begin{figure*}[!pt]
    \centering
        \begin{subfigure}[b]{0.33\textwidth}
            \centering
            \includegraphics[width=\linewidth, trim={0cm 0cm 0 0cm}, clip]{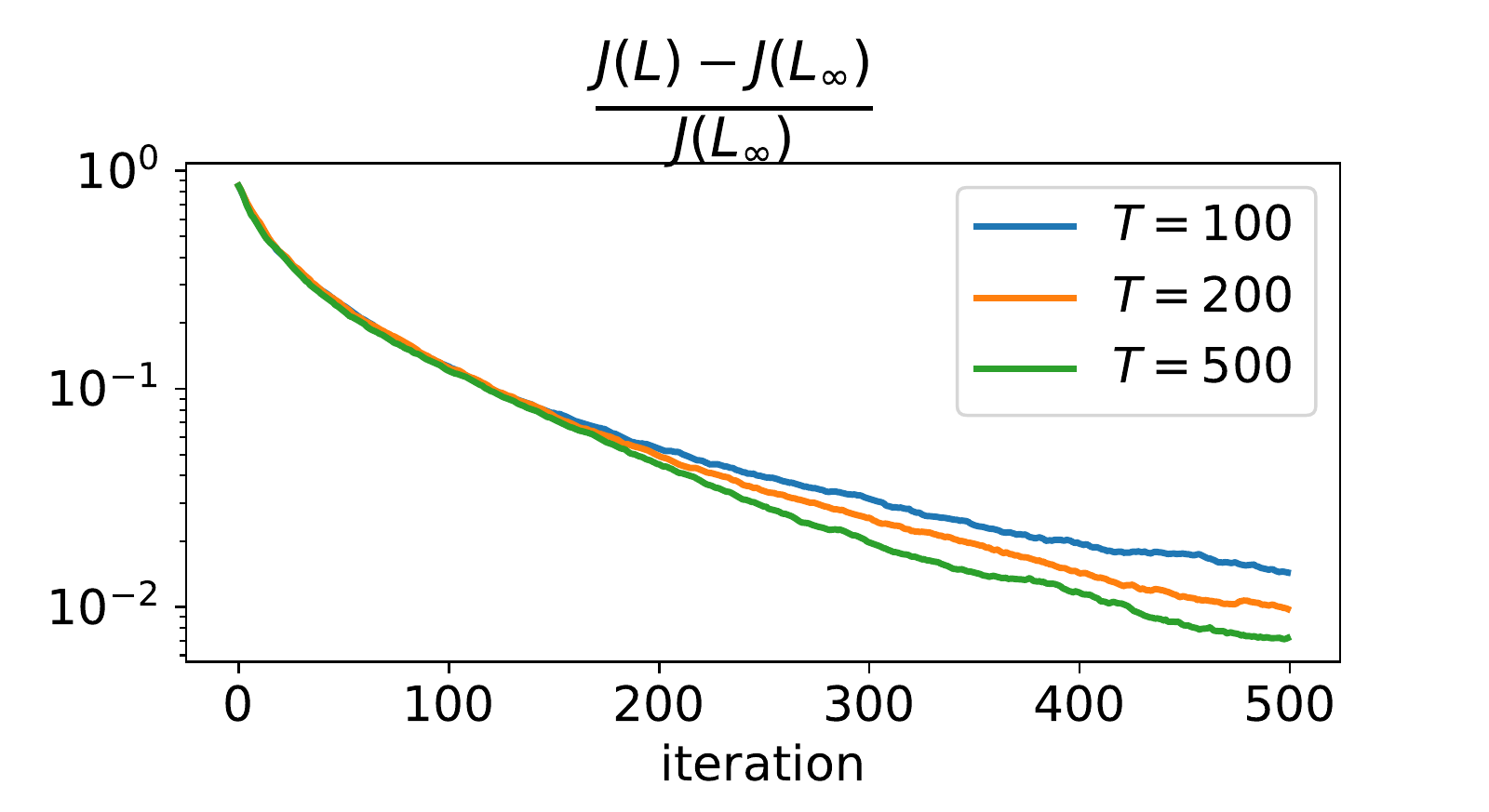}
             \caption{}
             \label{fig:1}
        \end{subfigure}
        \hfill
        \begin{subfigure}[b]{0.33\textwidth}
             \centering
             \includegraphics[width=\linewidth, trim={0cm 0cm 0 0cm}, clip]{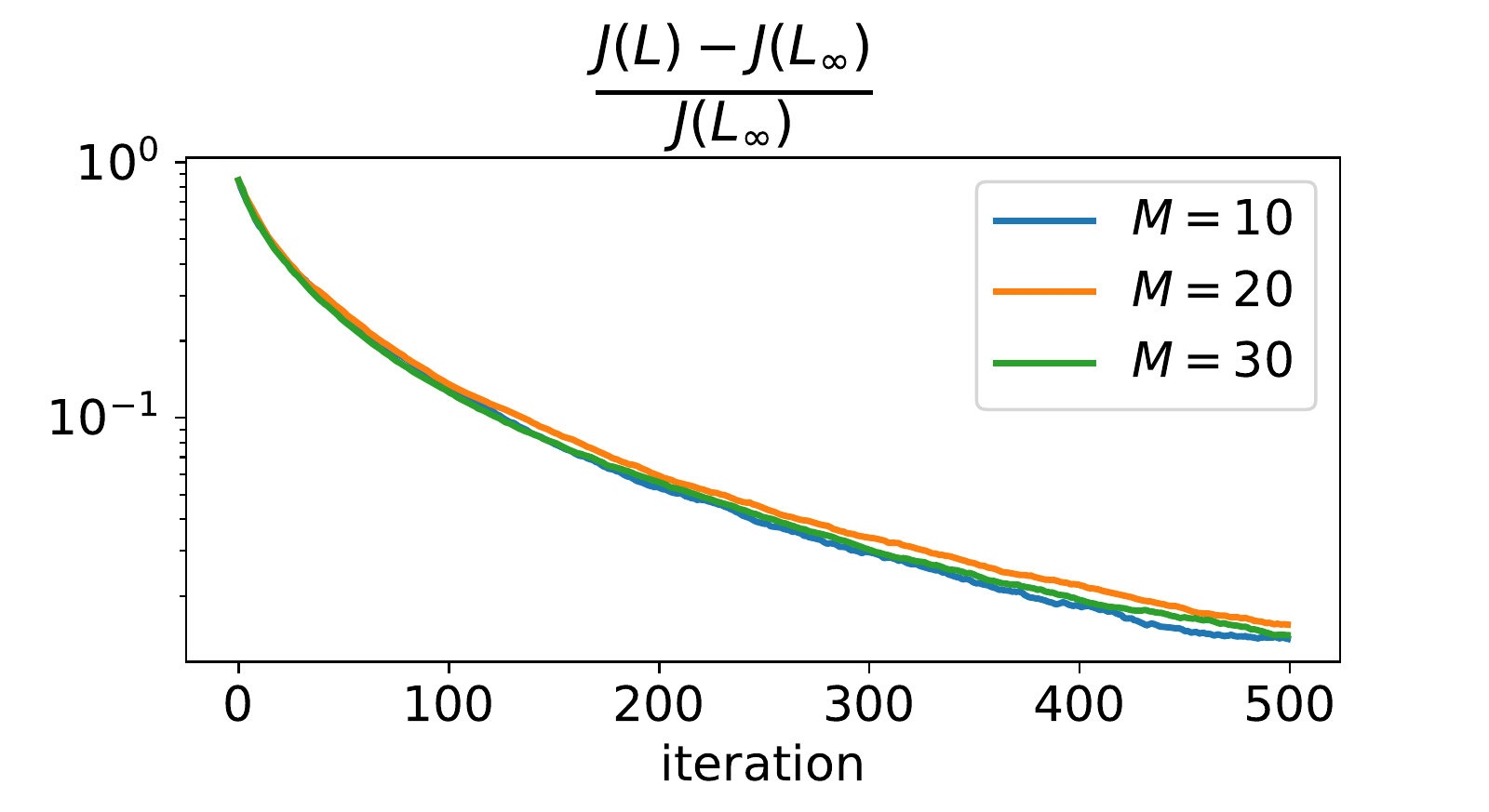}
             \caption{}
             \label{fig:2}
        \end{subfigure}
    \begin{subfigure}[b]{0.32\textwidth}
         \centering
         \pdfset{%
              20 w 
              0 J  
              0 j  
            }{\includegraphics[height= 1.3in, width=\linewidth, trim={0.8cm 0cm 0cm 0cm}, clip]{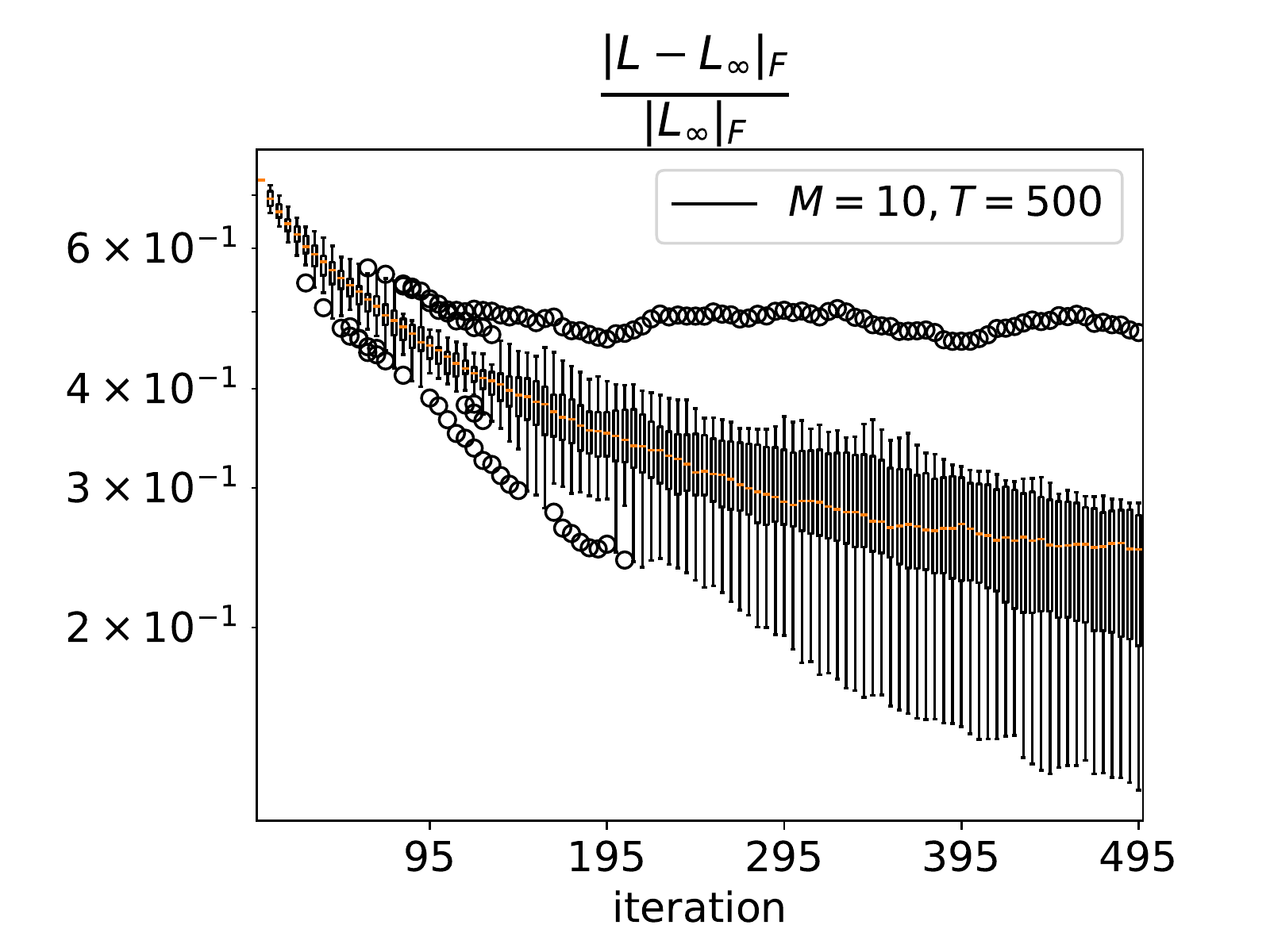}}
         \caption{}
         \label{fig:3}
    \end{subfigure}
    \caption{\small \ac{sgd} directly from output data and without prior knowledge of the noise covariances or state information. Mean progress of the normalized estimation error over 20 simulations obtained from data trajectories of a) different length $T$ and b) different batch size $M$; also, c) progress in the Kalman gain with the mean in orange, variance in black line and the outliers in circles.}
    \label{fig:sim}
    \vspace{-0.5cm}
\end{figure*}

\section{Conclusions}\label{sec:conclusion}
In this work, we considered the problem of learning the optimal Kalman gain with unknown process and measurement noise covariances. We proposed a direct stochastic \ac{po} algorithm with theoretical analysis that are based on the duality between optimal control and estimation. The extension for the other variant of the problem, where the dynamics/observation parameters are also (partially) unknown, is an immediate future direction of this work. 
\bibliographystyle{ieeetr}
\bibliography{citations}

\newpage
\appendix

\subsection{Proof of \Cref{prop:duality}}
\begin{proof}
By pairing the original state dynamics~\eqref{eqn:sysdyn} and its dual~\eqref{eqn:adjdyn}:
\begin{equation*}
    z(t+1)^\intercal x(t+1) - z(t)^\intercal x(t) = z(t+1)^\intercal \xi(t) + u(t+1)^\intercal H x(t). 
\end{equation*}
Summing this relationship from $t=0$ to $t=T-1$ yields,
\begin{gather*}
    z(T)^\intercal x(T) = z(0)^\intercal x(0) +  \sum_{t=0}^{T-1} z(t+1)^\intercal \xi(t) + u(t+1)^\intercal H x(t). 
\end{gather*}
Upon subtracting the estimate $a^\intercal \hat{x}_{\mathcal L}(T)$, and using the adjoint relationship~\eqref{eq:adjoint} and $z(T)=a$, it lead to  
\begin{align*}
    a^\intercal x(T) -& a^\intercal\hat{x}_{\mathcal L}(T)  = z(0)^\intercal x(0) - b^\intercal m_0 \\ &+ \textstyle\sum_{t=0}^{T-1} z(t+1)^\intercal \xi(t) - u(t+1)^\intercal w(t).
\end{align*}
Squaring both sides and taking the expectation concludes the first duality result.

The second claim follows from the identity
\begin{align*}
 &\E{}{\|y(T)\!-\!\hat{y}_{\mathcal L}(T)\|^2}  \!=  \!\E{}{\|Hx(T)\!-\!H\hat{x}_{\mathcal L}(T)\|^2 \!+\!\|w(T)\|^2}\\
 &\quad =\textstyle\sum_{i=1}^m \E{}{|H_i^\intercal x(T) - H_i^\intercal \hat{x}_{\mathcal L}(T)|^2} + \tr{R}, 
\end{align*}
and the application of the first result with $a=H_i$. 
\end{proof}

\subsection{Proof of \Cref{prop:coercive}}
\begin{proof}
Consider any $L \in \mathcal{S}$ and note that the right eigenvectors of $A$ and $A_L$ that are annihilated by $H$ are identical. Thus, by \ac{PBH} test, observability of $(A,H)$ is equivalent to observability of $(A_L,H)$. Therefore, there exists a positive integer $n_0 \leq n$ such that 
\begin{equation*}
    H_{n_0}^\intercal(L):=\begin{bmatrix}
    H^\intercal & A_L^\intercal H^\intercal & \ldots & (A_L^\intercal)^{n_0-1}H^\intercal
    \end{bmatrix}
\end{equation*}
is full-rank, implying that $H_{n_0}^\intercal(L) H_{n_0}(L)$ is positive definite.
Let $Q \succcurlyeq \underline{Q}\succ 0$, $R \succcurlyeq \underline{R} \succ 0$ uniformly in time for some matrices $\underline{Q}$ and $\underline{R}$.
Now, recall that for any such stabilizing gain $L$, we compute
\begin{gather*}
\begin{aligned}
    J(L) &= \tr{X_{\{-\infty:T\}}(L) H^\intercal H} \\
    &=  \tr{\sum_{t=0}^\infty (A_L)^t(Q_t+L R_t L^\intercal)(A_L^\intercal)^t H^\intercal H} \\
    &\geq \tr{\sum_{t=0}^\infty (A_L)^t(\underline{Q}+L \underline{R} L^\intercal)(A_L^\intercal)^t H^\intercal H} \\
    &= \tr{\sum_{t=0}^\infty  \sum_{k=0}^{n_0-1}(A_L)^{n_0t+k}(\underline{Q}+L \underline{R} L^\intercal)(A_L^\intercal)^{n_0t+k} H^\intercal H} \\
    & =  \tr{\sum_{t=0}^\infty  (A_L)^{n_0t}(\underline{Q}+L \underline{R} L^\intercal)(A_L^\intercal)^{n_0t} H_{n_0}^\intercal(L) H_{n_0}(L)}\\
    & \eqqcolon \tr{X_{n_0}(L) H_{n_0}^\intercal(L) H_{n_0}(L)},
\end{aligned}
\end{gather*}
where we used the cyclic property of trace and the inequality follows because for any PSD matrices $P_1 , P_2 \succcurlyeq 0$ we have
\begin{equation}\label{eq:postrace}
    \tr{P_1P_2} = \tr{P_2^{\frac{1}{2}} P_1 P_2^{\frac{1}{2}}} \geq 0.
\end{equation}
Also, $X_{n_0}(L)$ is well defined because $A_L$ is Schur stable if and only if $(A_L)^{n_0}$ is. Moreover, $X_{n_0}(L)$ coincides with the unique solution to the following Lyapunov equation
\[X_{n_0}(L) = (A_L)^{n_0} X_{n_0}(L) (A_L^\intercal)^{n_0} + \underline{Q} + L \underline{R} L^\intercal.\]
Next, as $\underline{Q}, \underline{R} \succcurlyeq 0$,
\begin{align}
    J(L) \geq& \lambdamin(H_{n_0}^\intercal(L) H_{n_0}(L)) \tr{X_{n_0}(L)} \nonumber\\
    \geq& \lambdamin(H_{n_0}^\intercal(L) H_{n_0}(L)) \tr{\sum_{t=0}^\infty  (A_L)^{n_0t} \underline{Q} (A_L^\intercal)^{n_0t}} \nonumber\\
    \geq& \lambdamin(H_{n_0}^\intercal(L) H_{n_0}(L)) \lambdamin(\underline{Q}) \sum_{t=0}^\infty  \tr{(A_L^\intercal)^{n_0t}(A_L)^{n_0t}} \nonumber\\
    \geq& \lambdamin(H_{n_0}^\intercal(L) H_{n_0}(L)) \lambdamin(\underline{Q}) \sum_{t=0}^\infty  \rho(A_L)^{2n_0t} , \label{eqn:Jlower}
\end{align}
where the last inequality follows by the fact that
\begin{multline*}
    \tr{(A_L^\intercal)^{n_0t}(A_L)^{n_0t}} = \|(A_L)^{n_0t}\|_F^2 \\
    \geq \|(A_L)^{n_0t}\|_{\text{op}}^2 \geq \rho((A_L)^{n_0t})^2 = \rho(A_L)^{2n_0t},
\end{multline*}
where $\|\cdot\|_F$ denotes the Frobenius norm and with $\|\cdot\|_{\text{op}}$ denoting the operator norm induced by 2-norm.
Now, by \Cref{lem:topology} and continuity of the spectral radius, as $L_k \to \partial \mathcal{S}$ we observe that $\rho(A_{L_k}) \to 1$. But then, the obtained lowerbound implies that $J(L_k) \to \infty$. On the other hand, as $\underline{Q}\succ 0$, $\underline{R} \succ 0$ are both time-independent, by using a similar technique we also provide the following lowerbound
\begin{align*}
    J(L) \geq& \tr{(\underline{Q}+L \underline{R} L^\intercal)\sum_{t=0}^\infty(A_L^\intercal)^{n_0t}  H_{n_0}^\intercal(L) H_{n_0}(L) (A_L)^{n_0t}}\\
    \geq& \lambdamin(H_{n_0}^\intercal(L) H_{n_0}(L)) \tr{\underline{Q} + L \underline{R} L^\intercal} \\
    \geq& \lambdamin(H_{n_0}^\intercal(L) H_{n_0}(L)) \tr{\underline{R} LL^\intercal}\\
    \geq& \lambdamin(H_{n_0}^\intercal(L) H_{n_0}(L)) \lambdamin(\underline{R}) \|L\|_F^2.
\end{align*}
Therefore, by equivalency of norms on finite dimensional spaces, $\|L_k\| \to \infty$ implies that $J(L_k) \to \infty$ which concludes that $J(.)$ is coercive on $\mathcal{S}$. 
Finally, note that for any $L \not\in \mathcal{S}$, by \cref{eqn:Jlower} we can argue that $J(L) = \infty$, therefore the sublevel sets $\mathcal{S}_\alpha \subset \mathcal{S}$ whenever $\alpha$ is finite. The compactness of $\mathcal{S}_\alpha$ is then a direct consequence of the coercive property and continuity of $J(.)$ (\Cref{lem:topology}).
\end{proof}

\subsection{Derivation of the gradient formula}
Next, we aim to compute the gradient of $J$ for the time-varying parameters.
For any admissible $\Delta$, we have
\begin{align*}
    &X_{\{-\infty:T\}}(L+\Delta)- X_{\{-\infty:T\}}(L) =\\
    &\sum_{t=1}^{\infty} \left(A_L\right)^{t} \left(  Q_t + L R_t  L^\intercal \right) \left(\star\right)^\intercal + \left(\star\right) \left(  Q_t + L R_t  L^\intercal \right) \left(A_L^\intercal\right)^{t} \\
    &\quad-\sum_{t=0}^{\infty} \left(A_L\right)^{t} \left(  \Delta R_t  L^\intercal + L R_t \Delta^\intercal \right)\left(A_L^\intercal\right)^{t} + o(\|\Delta\|),
\end{align*}
where the $\star$ is hiding the following term
\[\sum_{k=1}^t (A_L)^{t-k} \Delta H (A_L)^{k-1}.\]
Therefore, by linearity and cyclic permutation property of trace, we get that
\begin{gather*}
    \begin{aligned}
        &J(L+\Delta)- J(L) =\\
    &\tr{\Delta H\sum_{t=1}^{\infty}\sum_{k=1}^{t} 2\left(A_L\right)^{k-1}  \left(  Q_t + L R_t  L^\intercal \right)  \left(A_L^\intercal\right)^{t} H^\intercal H \left(A_L\right)^{t-k}}\\
    &\quad -\tr{\Delta \sum_{t=0}^{\infty} 2R_t  L^\intercal \left(A_L^\intercal\right)^{t} H^\intercal H \left(A_L\right)^{t}} + o(\|\Delta\|).
    \end{aligned}
\end{gather*}
Finally, by considering the Euclidean metric on real matrices induced by the inner product $\langle Q, P\rangle = \tr{Q^\intercal P}$, we obtain the gradient of $J$ as follows 
\begin{gather*}
\begin{aligned}
    \nabla & J(L) = -2\sum_{t=0}^{\infty} \left(A_L^\intercal\right)^{t} H^\intercal H  \left(A_L\right)^{t} L R_t \\
    +&2\sum_{t=1}^{\infty}\sum_{k=1}^{t}     \left(A_L^\intercal\right)^{t-k}  H^\intercal H\left(A_L\right)^{t} \left(  Q_t + L R_t  L^\intercal \right)  \left(A_L^\intercal\right)^{k-1} H^\intercal,
\end{aligned}
\end{gather*}
whenever the series are convergent!
And, by switching the order of the sums it simplifies to
\begin{gather*}
    \begin{aligned}
            \nabla & J(L) = -2\sum_{t=0}^{\infty} \left(A_L^\intercal\right)^{t} H^\intercal H  \left(A_L\right)^{t} L R_t\\ &+2\sum_{k=1}^{\infty}\sum_{t=k}^{\infty} \left[\left(A_L^\intercal\right)^{t-k}  H^\intercal H \left(A_L\right)^{t-k}\right] \\
            &\quad\cdot A_L \left[\left(A_L\right)^{k-1}\left(  Q_t + L R_t  L^\intercal \right)  \left(A_L^\intercal\right)^{k-1}\right] H^\intercal.\\
            =& -2\sum_{t=0}^{\infty} \left(A_L^\intercal\right)^{t} H^\intercal H  \left(A_L\right)^{t} L R_t +2\sum_{t=0}^{\infty} \left[\left(A_L^\intercal\right)^{t}  H^\intercal H \left(A_L\right)^{t}\right] \\
            &\quad\cdot A_L \left[\sum_{k=0}^{\infty}\left(A_L\right)^{k}\left(  Q_{t+k+1} + L R_{t+k+1}  L^\intercal \right)  \left(A_L^\intercal\right)^{k}\right] H^\intercal.
    \end{aligned}
\end{gather*} For the case of time-independent $Q$ and $R$, this reduces to
\begin{gather*}
\begin{aligned}
    \nabla  J(L) 
    =& -2Y_{(L)} L R \\
    & + 2Y_{(L)} A_L\left[\sum_{k=0}^{\infty} \left(A_L\right)^{k}\left(  Q + L R  L^\intercal \right)  \left(A_L^\intercal\right)^{k}\right] H^\intercal\\
    =& 2Y_{(L)} \left[-LR + A_L X_{(L)} H^{\intercal} \right].
\end{aligned}
\end{gather*}
where $Y_{(L)}=Y$ is the unique solution of
\[Y = A_L^\intercal Y A_L + H^\intercal H.\]

\subsection{Proof of the \Cref{prop:graddom}}
\begin{proof}
Note that $X=X_{(L)}$ satisfies 
\begin{equation}\label{eqn:X}
    X = A_{L} X A_{L}^\intercal + Q + L R L^\intercal.
\end{equation}
Then, by combining \cref{eqn:Xstar} and \cref{eqn:X}, and some algebraic manipulation, we recover part of the gradient information, i.e. $(-LR + A_L X H^\intercal)$, in the gap of cost matrices by arriving at the following identity
\begin{gather}\label{eq:X-X}
\begin{aligned}
    X&-X^* - A_{L^*} (X-X^*) A_{L^*}^\intercal 
    \\
    =&(LR - A_L X H^\intercal)(L-L^*)^\intercal 
    + (L-L^*)(R L^\intercal - H X A_L^\intercal) \\
    &- (L -L^*)R(L-L^*)^\intercal 
    - (L - L^*) H X H^\intercal (L - L^*)^\intercal \\
    \preccurlyeq& \frac{1}{a}(LR - A_L X H^\intercal) (R L^\intercal - H X A_L^\intercal)
    + a(L-L^*)(L-L^*)^\intercal\\
    &- (L - L^*) (R + H X H^\intercal) (L - L^*)^\intercal\\
\end{aligned}
\end{gather}
where the upperbound is valid for any choice of $a > 0$. Now, as $R\succ 0$, we choose $a = \lambdamin(R)/2$. As $X \succcurlyeq 0$, it further upperbounds
\begin{gather*}
\begin{aligned}
    X-X^* - A_{L^*} &(X-X^*) A_{L^*}^\intercal 
    \\
    \preccurlyeq& \frac{2}{\lambdamin(R)}(-LR + A_L X H^\intercal) (-R L^\intercal + H X A_L^\intercal) \\
    &- \frac{\lambdamin(R)}{2}(L - L^*) (L - L^*)^\intercal .
\end{aligned}
\end{gather*}
Now, let $\Tilde{X}$ and $\widehat{X}$ be, respectively, the unique solution of the following Lyapunov equations
\begin{align*}
    \Tilde{X} &= A_{L^*} \Tilde{X} A_{L^*}^\intercal + (-LR + A_L X H^\intercal) (-R L^\intercal + H X A_L^\intercal),\\
    \widehat{X} &= A_{L^*} \widehat{X} A_{L^*}^\intercal + (L - L^*) (L - L^*)^\intercal.
\end{align*}
Then by comparison, we conclude that 
\[X - X^* \preccurlyeq \frac{2}{\lambdamin(R)}\Tilde{X} - \frac{\lambdamin(R)}{2}\widehat{X}.\]
Recall that by the fact in \cref{eq:postrace},
\begin{multline}\label{eqn:ineqJ}
    J(L) - J(L^*) = \tr{(X-X^*)H^\intercal H} \\
    \leq \frac{2}{\lambdamin(R)}\tr{\Tilde{X}H^\intercal H} - \frac{\lambdamin(R)}{2}\tr{\widehat{X} H^\intercal H}.
\end{multline}
Let $Y^* \succ 0$ be the unique solution of
\[Y^* = A_{L^*}^\intercal Y^* A_{L^*} + H^\intercal H,\]
then, by cyclic permutation property
\begin{gather}
    \tr{\Tilde{X}H^\intercal H} = \tr{(-LR + A_L X H^\intercal) (-R L^\intercal + H X A_L^\intercal)Y^*} \nonumber\\
    \leq \frac{\lambdamax(Y^*)}{\lambdamin^2(Y_{(L)})} \tr{Y_{(L)}(-LR + A_L X H^\intercal) (-R L^\intercal + H X A_L^\intercal)Y_{(L)}} \nonumber\\
    = \frac{\lambdamax(Y^*)}{4\lambdamin^2(Y_{(L)})} \langle \nabla J(L), \nabla J(L) \rangle \label{eqn:ineqgrad}
\end{gather}
where the inequality follows by \cref{eq:postrace} and the last equality follows by the formula for the gradient $\nabla J(L)$.
Similarly, we obtain that
\begin{multline}\label{eq:Xhatbound}
    \tr{\widehat{X}H^\intercal H} = \tr{(L - L^*) (L - L^*)^\intercal Y^*} \\
    \geq \lambdamin(Y^*)\|L-L^*\|_F^2.
\end{multline}
Notice that the mapping $L \to Y_{(L)}$ is continuous on $\mathcal{S} \supset \mathcal{S}_\alpha$, and also by observability of $(A,H)$, $Y_{(L)} \succ 0$ for any $L \in \mathcal{S}$. To see this, let $H_{n_0}(L) \succ 0$ be as defined in \Cref{prop:coercive}. Then,
\begin{gather*}
\begin{aligned}
    Y_{(L)} &= \sum_{t=0}^\infty (A_L^\intercal)^t (H^\intercal H) (A_L)^t \\
    &= \sum_{t=0}^\infty  \sum_{k=0}^{n_0-1}(A_L^\intercal)^{n_0t+k} (H^\intercal H) (A_L)^{n_0t+k} \\
    &= \sum_{t=0}^\infty (A_L^\intercal)^{n_0t} H_{n_0}^\intercal(L) H_{n_0}(L) (A_L)^{n_0t}\\
    &\succcurlyeq H_{n_0}^\intercal(L) H_{n_0}(L) \succ 0.
\end{aligned}
\end{gather*}
Now, by \Cref{prop:coercive},  $\mathcal{S}_\alpha$ is compact and therefore we claim that the following infimum is attained with some positive value $\kappa_\alpha$:
\begin{equation}\label{eqn:kappa}
    \inf_{L \in \mathcal{S}_\alpha} \lambdamin(Y_{(L)}) \eqqcolon \kappa_\alpha >0.
\end{equation}
Finally, the first claimed inequality follows by combining the inequalities \cref{eqn:ineqJ}, \cref{eqn:ineqgrad} and \cref{eq:Xhatbound}, with the following choice of parameters
\[c_1(\alpha) \coloneqq \frac{2\lambdamin(R)}{\lambdamax(Y^*)}\kappa_\alpha^2, \quad\text{and}\quad c_2(\alpha) \coloneqq \frac{\lambdamin(Y^*)\lambdamin(R)^2 }{\lambdamax(Y^*)}\kappa_\alpha^2.\]
For the second claimed inequality, one arrives at the following identity by a computation similar to \cref{eq:X-X}:
\begin{gather*}
\begin{aligned}
    X&-X^* - A_{L} (X-X^*) A_{L}^\intercal 
    \\
    =&(L^*R - A_{L^*} X^* H^\intercal)(L-L^*)^\intercal 
    + (L-L^*)(R {L^*}^\intercal - H X^* A_{L^*}^\intercal) \\
    &+ (L -L^*)R(L-L^*)^\intercal 
    + (L - L^*) H X^* H^\intercal (L - L^*)^\intercal \\
    =& (L - L^*) (R + H X^* H^\intercal) (L - L^*)^\intercal
\end{aligned}
\end{gather*}
where the second equality follows because $Y_{(L)}\succ 0$ and thus
\[L^*R - A_{L^*} X^* H^\intercal = -Y_{(L)}^{-1}\nabla J(L^*) =0.\]
Recall that
\[J(L) - J(L^*) = \tr{(X-X^*)H^\intercal H},\]
then by the equality in \cref{eq:X-X} and cyclic property of trace we obtain
\[J(L) - J(L^*) = \tr{Z Y_{(L)}},\]
where
\begin{gather*}
\begin{aligned}
    Z \coloneqq& (L - L^*) (R + H X^* H^\intercal) (L - L^*)^\intercal\\
    \succcurlyeq& \lambdamin(R) (L - L^*) (L - L^*)^\intercal.
\end{aligned}
\end{gather*}
Therefore, for any $L \in \mathcal{S}_\alpha$, we have
\[J(L) - J(L^*) \geq \lambdamin(Y_{(L)})\tr{Z} \geq \lambdamin(R) \kappa_\alpha \|L - L^*\|_F^2,\]
and thus, we complete the proof by the following choice of parameter
\[c_3(\alpha) \coloneqq \lambdamin(R) \kappa_\alpha.\]
\end{proof}

\subsection{\acf{gf}}
In this section, we consider a policy update according to the  the \ac{gf} dynamics:
\begin{flalign*} 
\text{[\ac{gf}]} \qquad\qquad\qquad\quad \dot{L}_s = - \nabla J(L_s). &&
\end{flalign*}
We summarize the convergence result in the following Proposition which is a direct consequence of \Cref{prop:graddom} and we provide the proof for completeness.

\begin{proposition}\label{prop:gradflow}
Consider any sublevel set $\mathcal{S}_\alpha$ for some $\alpha >0$. Then, for any initial policy $L_0 \in \mathcal{S}_\alpha$, the \ac{gf} updates converges to optimality at a linear rate of $\exp(-c_1(\alpha))$ (in both the function value and the policy iterate). In particular, we have
\[J(L_s) -J(L^*) \leq (\alpha - J(L^*)) \exp(-c_1(\alpha)s),\]
and 
\[\|L_s - L^*\|_F^2  \leq \frac{\alpha - J(L^*)}{c_3(\alpha)} \exp(-c_1(\alpha) s).\]
\end{proposition}
\begin{proof}
Consider a Lyapunov candidate function $V(L) \coloneqq J(L) -J(L^*)$. Under the \ac{gf} dynamics
\[\dot{V}(L_s) = - \langle \nabla J(L_s), \nabla J(L_s) \rangle \leq 0.\]
Therefore, $L_s \in \mathcal{S}_\alpha$ for all $s>0$. But then, by \Cref{prop:graddom}, we can also show that
\[\dot{V}(L_s) \leq - c_1(\alpha) V(L_s) - c_2(\alpha) \|L_s -L^*\|_F^2, \quad \text{ for } s>0.\]
By recalling that $c_1(\alpha)> 0$ is a positive constant independent of $L$, we conclude the following exponential stability of the \ac{gf}:
\[V(L_s) \leq V(L_0)\exp(-c_1(\alpha) s) ,\]
for any $L_0 \in \mathcal{S}_\alpha$ which, in turn, guarantees convergence of $J(L_s) \to J(L^*)$ at the linear rate of $\exp(-c_1(\alpha))$.
Finally, the linear convergence of the policy iterates follows directly from the second bound in \Cref{prop:graddom}: 
\[\textstyle \|L_s - L^*\|_F^2 \leq \frac{1}{c_3(\alpha)} V(L_s) \leq \frac{V(L_0)}{c_3(\alpha)} \exp(-c_1(\alpha) s).\]
The proof concludes by noting that $V(L_0) \leq \alpha - J(L^*)$ for any such initial value $L_0 \in \mathcal{S}_\alpha$.
\end{proof}

\subsection{Proof of \Cref{lem:lipschitz}}
\begin{proof} 
Notice that the mappings $L \to X_{(L)}$, $L \to Y_{(L)}$ and $L \to A_{L}$ are all real-analytic on the open set $\mathcal{S} \supset \mathcal{S}_\alpha$, and thus so is the mapping $L \to \nabla J(L) = 2Y_{(L)} \left[-LR + A_L X_{(L)} H^{\intercal} \right]$. Also, by \Cref{prop:coercive},  $\mathcal{S}_\alpha$ is compact and therefore the mapping $L \to \nabla J(L)$ is $\ell$-Lipschitz continuous on $\mathcal{S}_\alpha$ for some $\ell =\ell(\alpha)>0$. In the rest of the proof, we attempt to characterize $\ell(\alpha)$ in terms of the problem parameters.
By direct computation we obtain
\begin{gather*}
\begin{aligned}
    &\nabla J(L_1) - \nabla J(L_2) =(2Y_{(L_1)} - 2Y_{(L_2)}) \left[-L_1R + A_{L_1} X_{(L_1)} H^{\intercal} \right] \\
    &+ 2Y_{(L_2)} \left(\left[-L_1R + A_{L_1} X_{(L_1)} H^{\intercal} \right] -\left[-L_2R + A_{L_2} X_{(L_2)} H^{\intercal} \right]\right)\\
    &= 2(Y_{(L_1)} - Y_{(L_2)}) \left[-L_1(R + H X_{(L_1)} H^{\intercal}) + A X_{(L_1)} H^{\intercal} \right]\\
    &+ 2Y_{(L_2)} \big[(L_2 - L_1)(R + H X_{(L_1)} H^{\intercal}) + A_{L_2} ( X_{(L_1)}  - X_{(L_2)}) H^{\intercal} \big].
\end{aligned}
\end{gather*}
Therefore,
\begin{multline}\label{eqn:grad-lip1}
    \|\nabla J(L_1) - \nabla J(L_2)\|_F^2 \leq
    \ell_1^2\|Y_{(L_1)} - Y_{(L_2)}\|_F^2 \\+ \ell_2^2\|L_1 - L_2\|_F^2 + \ell_3^2\|X_{(L_1)}  - X_{(L_2)}\|_F^2
\end{multline}
where 
\begin{align*}
    \ell_1 &= \ell_1(L_1) \coloneqq 2\|-L_1(R + H X_{(L_1)} H^{\intercal}) + A X_{(L_1)} H^{\intercal}\|_{\text{op}},\\
    \ell_2 &= \ell_2(L_1,L_2) \coloneqq 2 \|Y_{(L_2)}\|_{\text{op}}\, \|R + H X_{(L_1)} H^{\intercal}\|_{\text{op}},\\
    \ell_3 &= \ell_3(L_2) \coloneqq 2 \|Y_{(L_2)}\|_{\text{op}}\, \|A_{L_2}\|_{\text{op}}\, \|H^{\intercal}\|_{\text{op}}.
\end{align*}
On the other hand, by direct computation we obtain
\begin{gather}\label{eqn:gapY}
\begin{aligned}
    Y_{(L_1)}& - Y_{(L_2)} - A_{L_1}^\intercal (Y_{(L_1)} - Y_{(L_2)}) A_{L_1} \\
     =&(L_2-L_1)^\intercal H^\intercal  Y_{(L_2)} A_{L_2} + A_{L_2}^\intercal Y_{(L_2)} H (L_2-L_1) \\
    &+ (L_1-L_2)^\intercal H^\intercal  Y_{(L_2)} H (L_1-L_2) \\
    \preccurlyeq&  \|L_1-L_2\|_F \, \ell_4 \, I
\end{aligned}
\end{gather}
where
\begin{multline*}
    \ell_4 = \ell_4(L_1,L_2) \coloneqq  2\|
     H^\intercal  Y_{(L_2)} A_{L_2}\|_{\text{op}} \\
     + \|H^\intercal  Y_{(L_2)} H (L_1-L_2)\|_{\text{op}}.
\end{multline*}
Now, consider the mapping $L \to Z_{(L)}$ where $Z_{(L)} = Z$ is the unique solution of the following Lyapunov equation:
\[Z = A_L^\intercal Z A_L + I,\]
which is well-defined and continuous on $\mathcal{S}\supset \mathcal{S}_\alpha$. Therefore, by comparison, we claim that
\[\|Y_{(L_1)} - Y_{(L_2)}\|_F \preccurlyeq \|L_1-L_2\|_F \; \ell_4 \; \|Z_{(L_1)}\|_F.\]
By a similar computation to that of \cref{eq:X-X}, we obtain that
\begin{gather}\label{eqn:gapX}
\begin{aligned}
    X_{(L_1)}-X_{(L_2)} - &A_{L_2} (X_{(L_1)}-X_{(L_2)}) A_{L_2}^\intercal 
    \\
    =&(L_1R - A_{L_1} X_{(L_1)} H^\intercal)(L_1-L_2)^\intercal \\
    &+ (L_1-L_2)(R L_1^\intercal - H X_{(L_1)} A_{L_1}^\intercal) \\
    &- (L_1 -L_2)R(L_1-L_2)^\intercal \\
    &- (L_1 - L_2) H X_{(L_1)} H^\intercal (L_1 - L_2)^\intercal \\
    \preccurlyeq& \|L_1-L_2\|_F \, \ell_5 \, (Q + L_2 R L_2^\intercal)
\end{aligned}
\end{gather}
where
\begin{equation*}
    \ell_5 = \ell_5(L_1) \coloneqq  2\|-L_1R + A_{L_1} X_{(L_1)} H^\intercal\|_{\text{op}}/ \lambdamin(Q).
\end{equation*}
Therefore, by comparison, we claim that
\[\|X_{(L_1)} - X_{(L_2)}\|_F \preccurlyeq \|L_1-L_2\|_F \; \ell_5 \; \|X_{(L_2)}\|_F.\]
%
Finally, by compactness of $\mathcal{S}_\alpha$, we claim that the following supremums are attained and thus, are achieved with some \emph{finite} positive values:
\begin{align*}
    & \bar\ell_1(\alpha) \coloneqq \sup_{L_1,L_2\in\mathcal{S}_\alpha} \ell_1(L_1) \ell_4(L_1,L_2) \; \|Z_{(L_1)}\|_F,\\
    & \bar\ell_2(\alpha) \coloneqq \sup_{L_1,L_2\in\mathcal{S}_\alpha} \ell_2(L_1,L_2),\\
    & \bar\ell_3(\alpha) \coloneqq \sup_{L_1,L_2\in\mathcal{S}_\alpha} \ell_3(L_2) \ell_5(L_1) \|X_{(L_2)}\|_F.
\end{align*}
Then, the claim follows by combining the bound in \cref{eqn:grad-lip1} with \cref{eqn:gapY} and \cref{eqn:gapX},  and the following choice of 
\begin{equation*}
    \ell(\alpha) \coloneqq \sqrt{\bar{\ell}_1^2(\alpha) + \bar{\ell}_2^2(\alpha) + \bar{\ell}_3^2(\alpha)}.
\end{equation*}
\end{proof}

\subsection{Proof of \Cref{thm:graddescent }}
\begin{proof}
First, we argue that the \ac{gd} update with such a step size does not leave the initial sublevel set $\mathcal{S}_\alpha$ for any initial $L_0 \in \mathcal{S}_\alpha$.
In this direction, consider $L(\eta) = L_0 - \eta \nabla J(L_0)$ for $\eta \geq 0$ where $L_0 \neq L^*$. 
Then, by compactness of $\mathcal{S}_\alpha$ and continuity of the mapping $\eta \to J(L(\eta))$ on $\mathcal{S} \supset \mathcal{S}_\alpha$, the following supremum is attained with a positive value $\bar\eta_0$:
\[\bar\eta_0 \coloneqq \sup \{\eta: J(L(\zeta)) \leq \alpha, \forall \zeta \in [0,\eta]\}, \]
where positivity of $\bar\eta_0$ is a direct consequence of the strict decay of $J(L(\eta))$ for sufficiently small $\eta$ as $\nabla J(L_0) \neq 0$.
This implies that $L(\eta) \in \mathcal{S}_\alpha \subset \mathcal{S}$ for all $\eta \in [0,\bar\eta_0]$ and $J(L(\bar\eta_0)) = \alpha$. Next, by the Fundamental Theorem of Calculus and smoothness of $J(\cdot)$ (\Cref{lem:topology}), for any $\eta \in [0,\bar\eta_0]$ we have that,
\begin{align*}
    J(L(\eta)) &-J(L_0) - \langle \nabla J(L_0), L(\eta) - L_0\rangle\\
    &= \int_0^1 \langle \nabla J(L(\eta s)) -\nabla J(L_0), L(\eta) - L_0\rangle d s\\
    &\leq \|L(\eta) - L_0\|_F \int_0^1 \|\nabla J(L(\eta s)) -\nabla J(L_0)\|_F  d s\\
    &\leq \ell(\alpha) \|L(\eta) - L_0\|_F \int_0^1 \|L(\eta s) -L_0\|_F  d s\\
    &= \frac{1}{2}\ell(\alpha)\eta \|L(\eta) - L_0\|_F \|\nabla J(L_0)\|_F,
\end{align*}
where $\|\cdot\|_F$ denotes the Frobenius norm, the first inequality is a consequence of Cauchy-Schwartz, and the second one is due to \Cref{lem:lipschitz} and the fact that $L(\eta s)$ remains in $\mathcal{S}_\alpha$ for all $s \in [0,1]$.\footnote{Note that a direct application of Descent Lemma \cite[Lemma 5.7]{beck2017frist} may not be justified as one has to argue about the uniform bound for the Hessian of $J$ over the non-convex set $\mathcal{S}_\alpha$ where $J$ is $\ell(\alpha)$-Lipschitz only on $\mathcal{S}_\alpha$. Also see the proof of \cite[Theorem 2]{mohammadi2021convergence}.} By the definition of $L(\eta)$, it now follows that,
\begin{align}\label{eq:decayineq}
    J(L(\eta)) -J(L_0) \leq \eta \|\nabla J(L_0)\|_F^2 \left( \frac{\ell(\alpha) \eta}{2} - 1 \right).
\end{align}
This implies $J(L(\eta)) \leq J(L_0)$ for all $\eta \leq 2/\ell(\alpha)$, and thus concluding that $\bar\eta_0 \geq 2/\ell(\alpha)$. This justifies that $L(\eta) \in \mathcal{S}_\alpha$ for all $\eta \in [0, 2/\ell(\alpha)]$.
Next, if we consider the \ac{gd} update with any fixed stepsize $\eta \in (0,1/\ell(\alpha)]$ and apply the bound in \cref{eq:decayineq} and the gradient dominance property in \Cref{prop:graddom}, we obtain
\[\textstyle J(L_1) - J(L_0) \leq  \eta c_1 (\frac{\ell(\alpha) \eta}{2}-1)[J(L_0) - J(L^*)],\]
which by subtracting $J(L^*)$ results in
\[\textstyle J(L_1) - J(L^*) \leq  \left(1- \frac{\eta c_1}{2} \right)[J(L_0) - J(L^*)],\]
as $\eta c_1 ({\ell(\alpha) \eta}/{2}-1) \leq -{\eta c_1}/{2}$ for all $\eta \in (0,1/\ell(\alpha)]$.
By induction, and the fact that both $c_1(\alpha)$ and the choice of $\eta$ only depends on the value of $\alpha$, we conclude the convergence in the function value at a linear rate of $1- (\eta c_1/2)$ and the constant coefficient of $\alpha - J(L^*) \geq J(L_0) - J(L^*)$.
To complete the proof, the linear convergence of the policy iterates  follows directly from the second bound in \Cref{prop:graddom}.
\end{proof}

\subsection{Proof of \Cref{lem:grad-approx}}
\begin{proof}
For small enough $\Delta \in \mathbb R^{n\times m}$,
\begin{multline*}
   \varepsilon(L+\Delta,\mathcal Y) - \varepsilon(L,\mathcal Y) =  \|e_{T}(L + \Delta) \|^2 - \|e_{T}(L) \|^2 =\\
    2 \tr{(e_{T}(L + \Delta) - e_{T}(L) )e_{T}^\intercal(L)} + o(\|\Delta\|)).
\end{multline*}
The difference 
\begin{gather*}
    e_{T}(L+\Delta) - e_{T}(L) = E_1(\Delta) + E_2(\Delta) + o(\|\Delta\|),
\end{gather*}
with the following terms that are linear in $\Delta$:
\begin{align*}
    E_1(\Delta) &\coloneqq \textstyle -\sum_{t=0}^{\infty} H  (A_L)^{t} \Delta y(T-t-1),\\
    E_2(\Delta) &\coloneqq \textstyle \sum_{t=1}^{\infty}\sum_{k=1}^{t} H  (A_L)^{t-k} \Delta H (A_L)^{k-1} L y(T-t-1) .
\end{align*}
Therefore, combining the two identities, the definition of gradient under the inner product $\langle A,B\rangle :=\tr{AB^\intercal}$, and ignoring the higher order terms in $\Delta$ yields,
\begin{gather*}
    \langle \nabla_L \varepsilon(L,y), \Delta \rangle =  2 \tr{(E_1(\Delta)+ E_2(\Delta) )e_{T}^\intercal(L)},
\end{gather*}
which by linearity and cyclic permutation property of trace reduces to:
\begin{gather*}
    \langle \nabla_L \varepsilon(L,y), \Delta \rangle = - 2 \tr{\Delta \left(\sum_{t=0}^{\infty} y(T-t-1) e_{T}^\intercal(L) H (A_L)^{t} \right)}\\
    + 2 \tr{\Delta \left(\sum_{t=1}^{\infty}\sum_{k=1}^{t} H (A_L)^{k-1} L y(T-t-1) e_{T}^\intercal(L) H (A_L)^{t-k} \right)}.
\end{gather*}
This holds for all admissible $\Delta$, concluding the formula for the gradient.
\end{proof}

\end{document}